\documentclass[onecolumn, a4paper]{IEEEtran}

\usepackage[normalem]{ulem} 

\usepackage{array}
\usepackage{cite}
\usepackage{tabu}
\usepackage{color}

\usepackage[normalem]{ulem}

\makeatletter
\newcommand{\thickhline}{%
    \noalign {\ifnum 0=`}\fi \hrule height 1pt
    \futurelet \reserved@a \@xhline
}
\newcolumntype{"}{@{\hskip\tabcolsep\vrule width 1pt\hskip\tabcolsep}}
\makeatother

\usepackage{hyperref}
\hypersetup{citecolor=blue,pdfborder=0 0 0,pdfstartview={FitH}}
\usepackage{amsfonts}
\usepackage{tabularx}

\ifCLASSINFOpdf
\usepackage[pdftex]{graphicx}
\else
\fi

\usepackage{caption}
\usepackage{subcaption}

\usepackage{amsmath, amsthm, amssymb}

\newtheorem{theorem}{Theorem}
\newtheorem{lemma}[theorem]{Lemma}
\newtheorem{example}{Example}
\newtheorem{proposition}[theorem]{Proposition}
\newtheorem{definition}{Definition}

\newtheorem{cnst}{Construction}

\newcommand{\vlen}{n}
\newcommand{\blen}{k}
\newcommand{\vbr}{t}
\newcommand{\vz}{\boldsymbol{z}}
\newcommand{\vy}{\boldsymbol{y}}
\newcommand{\vx}{\boldsymbol{x}}
\newcommand{\vs}{\boldsymbol{s}}
\newcommand{\vv}{\boldsymbol{v}}
\newcommand{\vu}{\boldsymbol{u}}
\newcommand{\vp}{\boldsymbol{p}}
\newcommand{\va}{\boldsymbol{a}}
\newcommand{\vb}{\boldsymbol{b}}
\newcommand{\ve}{\boldsymbol{e}}
\newcommand{\vc}{\boldsymbol{c}}
\newcommand{\repZ}{Z}
\newcommand{\enc}{ENC}
\newcommand{\dmax}{D_{p}}
\newcommand{\pmax}{p'}
\newcommand\Hbl[2]{B_{#1}(#2)}
\newcommand\Hbln[1]{B_{#1}}
\newcommand{\code}{\mathcal{C}}
\newcommand{\mH}{\mathsf{H}}
\newcommand{\mS}{\mathsf{S}}
\newcommand{\mI}{\mathsf{I}}
\newcommand{\dS}{\mathsf{\tilde{S}}}
\newcommand{\dr}{\tilde{r}}
\newcommand{\enci}{\mathcal{E}}
\newcommand{\deci}{\mathcal{D}}
\newcommand{\reso}{m}
\newcommand\enum[2]{\mathrm{enum}_{#1}(#2)}
\newcommand{\vzero}{\boldsymbol{0}}
\newcommand{\detu}{U}
\newcommand{\rbits}{\boldsymbol{r}}
\newcommand{\rat}{r}
\newcommand{\ssz}{s}

\definecolor{purp}{rgb}{0.5,0,0.5}
\definecolor{grn}{rgb}{0.2,0.6,0.3}

\newcommand\ycr[1]{{\color{black}#1}}
\newcommand\ycrr[1]{{\color{black}#1}}

\begin{document}

%

\title{Efficient Compression of Long Arbitrary Sequences with No Reference at the Encoder}

\author{Yuval~Cassuto,~\IEEEmembership{Senior Member,~IEEE,} and Jacob~Ziv,~\IEEEmembership{Life Fellow,~IEEE} 
\thanks{Parts of this work were presented at the 2020 International Zurich Seminar on Information and Communication.}
\thanks{Authors are with the Andrew and Erna Viterbi Department of Electrical Engineering, Technion - Israel Institute of Technology, Haifa, Israel 3200003 (emails: \{ycassuto,jz\}@ee.technion.ac.il).}}

%



\maketitle

\begin{abstract}
In a distributed information application an encoder compresses an arbitrary vector while a similar reference vector is available to the decoder as side information. For the Hamming-distance similarity measure, and when guaranteed perfect reconstruction is required, we present two contributions to the solution of this problem. One result shows that when a set of potential reference vectors is available to the encoder, lower compression rates can be achieved when the set satisfies a certain clustering property. Another result reduces the best known decoding complexity from exponential in the vector length $\vlen$ to $O(\vlen^{1.5})$ by generalized concatenation of inner coset codes and outer error-correcting codes. One potential application of the results is the compression of DNA sequences, where similar (but not identical) reference vectors are shared among senders and receivers.       
   
\end{abstract}

\section{Introduction}\label{sec:intro}
Data compression exploits similarity between data to save transmission bandwidth or storage. Similarity can be internal to one data sequence, or external between multiple data sequences. In classical information theory, similarity is modeled through the abstraction of an {\em information source}, which is defined probabilistically~\cite{ShannonC:48}. Intra-sequence similarity exists because a long sequence is extracted from a source with a given probability distribution, and inter-sequence similarity is due to a non-trivial joint distribution between the sources that generate the sequences. Often times it is challenging to define the information source by a probability distribution. Such is the case, for example, in DNA sequences that are generated by nature, with a distribution that is unclear and hard to define. Still, compressing long sequences from unstructured sources is highly desired with the advent of data-rich applications, which generate, analyze, and manipulate volumes of these sequences.

In this paper we study and develop tools for compression of sequences lacking probabilistic models. The setup of interest is compressing at the encoder a sequence (vector) $\vy$ that is similar to a reference vector $\vz$ available at the decoder, while similarity is expressed by a bound on the Hamming distance between $\vy$ and $\vz$. Our particular contributions to this setup are in two directions: first is a theoretical study of the case where the encoder has a set of candidate reference vectors, but does not know which particular $\vz$ from the set the decoder has; second is low-complexity compression and decompression for guaranteed zero-error reconstruction.  

An encoder compressing a vector $\vy$ for a decoder having {\em side-information} $\vz$ is a classical and well-studied problem in information theory. In particular, it is covered (as a special case) by the {\em Slepian-Wolf} coding scheme~\cite{SlepianW:73b} when the distributions of $\vy$ and $\vy$-given-$\vz$ are known. For cases when the distributions are unknown, Ziv~\cite{ZivJ:84} pursued the {\em individual-sequence} approach where statistical properties are replaced by combinatorial finite-state complexity measures. However, these combinatorial measures too are hard to characterize for general sequences of certain type, e.g., DNA sequences. This leaves us with the {\em Hamming distance} as the most rudimentary and robust measure of similarity between sequences. Compressing $\vy$ given side-information $\vz$ at the decoder, where $\vy$ and $\vz$ have bounded Hamming distance, was studied by Orlitsky and Viswanathan in~\cite{OrlitskyA:03}. They show a reduction of the Hamming-bounded compression problem to error-correcting codes in the Hamming metric, under the framework of {\em coset coding}. A similar scheme but for sets instead of sequences appears in~\cite{MinskyY:03}, and followed by extensions of the techniques motivated by biometric authentication~\cite{DodisY:08}. Many results exist, starting with~\cite{WynerA:74}, that apply coset coding to source coding (see an extensive study in~\cite{PradhanS:03}), but the uniqueness of~\cite{OrlitskyA:03} is that zero-error reconstruction is {\em guaranteed}, as needed in the applications that drive our present study.

This paper continues the line of work on guaranteed-success compression with Hamming-bounded side information. In the first part of the paper (Section~\ref{sec:dmax}), we study the case where the encoder as usual does not know the decoder's reference vector $\vz$, but it does have a set $\repZ$ of vectors that contains $\vz$ (among many other vectors). Our results in this part show that if the vectors in $\repZ$ have a certain well-defined ``clustering'' property, then it is possible to reduce the compression rate below the best known. This can be achieved without any probabilistic assumptions on the set $\repZ$, and without directly enforcing a bound on its size. Our results in this part are for guaranteed-decoding {\em average} compression rate, where the average is taken over the random hash\footnote{hash functions are also known as binning functions in information theory.} function used, and {\em not} over the input $\vy$ (which has no probability distribution). For the same model our results also include a lower bound on compression rate for any scheme that uses random hashing. In the second part of the paper (Section~\ref{sec:low_complexity}), we return to the classical model of~\cite{OrlitskyA:03} (no $\repZ$ in the encoder), and propose coding schemes with low complexity of encoding and decoding. For guaranteed decoding of length-$\vlen$ vectors with a constant fractional distance bound $p$, existing schemes require decoding complexity that is exponential in $\vlen$ due to the complexity of decoding an error-correcting code. Our proposed schemes have $O(n\sqrt{n})$ decoding complexity, which is low enough for practical implementation even for long input sequences. For low distance fractions $p$, our scheme has low compression rates, although not as low as the prior schemes that do not consider the decoding complexity. We use codes with structure similar to {\em generalized concatenation} (GC) codes~\cite{BlokhZyablov:76,ZyablovV:99} -- in particular {\em generalized error-locating} (GEL) codes~\cite{ZyablovV:72,BossertBook:99}. Applying the GEL code concatenation for compression requires to combine inner coset codes with outer error-correcting codes, while in the known construction both inner and outer codes are error-correcting codes. Moreover, using the known decoding algorithms for GEL (and GC) codes~\cite{ZyablovV:75} results in total decoding complexity that is above quadratic in $\vlen$, thus we use lower-complexity decoders to get the desired $O(n\sqrt{n})$. Our results show that when the distance fractions $p$ are small, low compression rates are achieved, which thanks to the low complexity may offer an alternative to compression algorithms not using side information at all. If one lifts the requirement for guaranteed decoding, then existing work (e.g.~\cite{UyematsuT:01,SmithA:07}) using classical concatenation~\cite{ForneyD_book:66} can achieve lower compression rates. Uyematsu~\cite{UyematsuT:01} uses classical concatenation for Slepian-Wolf coding that succeeds with high probability over the source distribution, and Smith~\cite{SmithA:07} provides a scheme for compression with side information at the decoder that succeeds with high probability over the shared randomness between encoder and decoder (this capability is extended to Slepian-Wolf coding in~\cite{ChumbalovD:18}.)
 
The theoretical setups studied in this paper are general, and may find use in various data-rich distributed applications involving storage and communications. However, applications involving {\em DNA sequences} are a particular motivation for this study. DNA sequences are extremely long (hundreds of megabytes for full-genome sequences), and in emerging personal-medicine applications they are stored and communicated by various resource-limited entities. For DNA applications, the set $\repZ$ of candidate reference vectors in Section~\ref{sec:dmax} models similar sequences available in the sender's local storage. The scheme of Section~\ref{sec:low_complexity} with its low guaranteed-decoding complexity is motivated by the long lengths of DNA sequences, and the importance of their perfect reconstruction. Most current compression schemes for DNA sequences use a reference sequence {\em in the encoder} (see, e.g.,~\cite{BonfieldJ:13,ZhangY:15}), and are thus forced to use generic reference vectors with weak similarity to the compressed vector $\vy$. Freeing the encoder from having the reference vector allows compressing $\vy$ with a smaller distance parameter $p$, building on the many similar vectors the decoder has in its local storage. 

The advantage of applying the generalized-concatenation approach for compression is that different inner codes can in future work accommodate additional similarity measures, for example $\vy$ and $\vz$ differing by {\em insertions and deletions}.

\section{Problem Model and Definitions}\label{sec:model}
In the problem setup we consider, there is an input vector we wish to convey (transmit or store) under the assumption that the party requesting this vector has a ``similar'' vector as a side-information vector (also called reference vector in the sequel). ``Similar'' here refers to having a bounded Hamming distance from the input vector. A length-$\vlen$ vector $\vy$ is given as input to the {\em encoder}, which maps $\vy$ to a vector $\enc(\vy)$ such that the {\em decoder} will be able to perfectly reproduce $\vy$ from $\enc(\vy)$ given a vector $\vz$ that satisfies $d_H(\vy,\vz)\leq p \vlen$, where $0<p<1$ is a real-valued parameter and $d_H(\cdot,\cdot)$ is the standard Hamming distance between vectors. The vector $\vz$ at the decoder is {\em not} known to the encoder. An encoder+decoder pair is called a {\em coding scheme}. The objective is to find a coding scheme that minimizes $|\enc(\vy)|$, the number of bits in $\enc(\vy)$, where either the worst-case or average-case $|\enc(\vy)|$ will be of interest, and the average is taken with respect to the randomization used by the algorithms without assuming any probability distribution on $\vy$. In both the worst case and the average case the decoder must recover $\vy$ without error.

\subsection{New model: reference-vector \underline{set} known to encoder}\label{subsec:formal_def}
Let $\repZ=\{\vz_1,\ldots,\vz_M\}$ be a set of vectors, where each vector $\vz_i$ is a binary vector of length $\vlen$. The set $\repZ$ is known to the encoder, and it contains the reference vector $\vz$ available at the decoder. While the encoder knows $\repZ$, it does {\em not} know the specific $\vz$ that the decoder has. The situation that the encoder knows $\repZ$ (but not $\vz$) can be encountered in practice when the encoder has access to a large repository of reference vectors, some of which are available to the decoder (but not clear which exactly). 

\subsection{Structured reference vectors: the $p$-spread parameter}\label{subsec:struct_SI}
Throughout the paper we will generally consider the set $\repZ$ of reference vectors as general and arbitrary, and in particular not assumed to have any stochastic properties. One useful parameter to characterize $\repZ$ is $\dmax$ we define next.
\begin{definition}
Given a set $\repZ$ of reference vectors we define $\dmax$ as
\begin{equation} \dmax(\repZ) \triangleq \max_{\vz_i,\vz_j:d_H(\vz_i,\vz_j)\leq 2p\vlen}d_H(\vz_i,\vz_j).\label{eq:dmax}\end{equation}
In words, $\dmax(\repZ)$ is the maximal distance between a pair of vectors in $\repZ$ whose distance is at most $2p\vlen$.
\end{definition}
Note that for any $\repZ$ we have the upper bound $\dmax(\repZ)\leq 2p\vlen$. When this upper bound is strict, it means that the set $\repZ$ has a ``clustering'' property, where vectors that are in the same neighborhood (have distance $\leq 2p\vlen$) are not very far from each other (have distance $\leq \dmax<2p\vlen$). For convenience, we define the {\em $p$-spread parameter} $\pmax$ of $\repZ$ as
\begin{equation} \pmax(\repZ,p) \triangleq \frac{\dmax(\repZ)}{2n}.\label{eq:pmax}\end{equation}
Later in the paper we will omit the arguments $\repZ$ and $p$ that are clear from the context, and just use $\pmax$. With this notation we have the upper bound \[\pmax \leq p. \]

The definition of the $p$-spread parameter $p'$ introduces structure to the set  $\repZ$. When $p'=p$ the vectors in $\repZ$ can be arbitrary, while $p'<p$ implies that the vectors in $\repZ$ are more ``clustered'' in the sense that pairs are either close $d_H(\vz_i,\vz_j)\leq 2p'$ or far $d_H(\vz_i,\vz_j) > 2p$, with a forbidden distance range in between. The $p$-spread parameter is the simplest combinatorial way we have found to model vector clustering, which is an important feature in applications like DNA compression. It is important to note that the $p$-spread parameter does {\em not} degenerate $\repZ$ to disjoint clusters of vectors with $d_H(\vz_i,\vz_j)\leq 2p'$, as seen in the next example.
\begin{example}\label{ex:p_spread}
For $\vlen=7$, consider the following example of $\repZ$.
\begin{equation}\repZ = \{ 0000000,0111000,1110000,1111000,1111111 \}.\label{eq:p_spread}\end{equation}
When $p=3/7$, we see that $\pmax(\repZ,p)=2/7$, because any two vectors in $\repZ$ that are at distance $6$ or less are also at distance $4$ or less. The set $\repZ$ models that both subsets $\{0000000,0111000,1110000,1111000\}$ and $\{0111000,1110000,1111000,1111111\}$ (which overlap) have some degree of similarity expressed in being at distance at most $4$ from each other. Because $0000000$ and $1111111$ are not similar according to this definition, they must be dissimilar in the sense of being at distance more than $6$ from each other.
\end{example}

\subsection{Hamming balls and anticodes}
In our results we define the proximity between input and reference vectors using the Hamming metric. Hence the following definitions will be useful. We denote by $\Hbl{l}{\vx}$ the {\em Hamming ball} of radius $l$ around the vector $\vx$, that is, $\Hbl{l}{\vx}=\{\vs\in\{0,1\}^\vlen:d_H(\vs,\vx)\leq l\}$. The size (number of vectors) of the Hamming ball is denoted $|\Hbl{l}{\vx}|$, and because it does not depend on the argument $\vx$ we denote it $|\Hbln{l}|$. We will use a well-known combinatorial inequality
\[ \forall \alpha<1/2,~|\Hbln{\alpha\vlen}| \leq 2^{\vlen H(\alpha)},  \]
where $H(\alpha)\triangleq -\alpha\log_2(\alpha)-(1-\alpha)\log_2(1-\alpha)$ is the binary entropy function.

We also use the definition of an anticode. A set of vectors $S\subset \{0,1\}^n$ is called an {\em anticode} of diameter $l$ if any two vectors $\vs_1,\vs_2\in S$ satisfy $d_H(\vs_1,\vs_2)\leq l$.

\subsection{Random hash functions}
A central tool in our proofs is {\em random hash functions}. A hash function $u:\{0,1\}^n\rightarrow \{0,1\}^m$ is a mapping from vectors of $n$ bits to vectors of $m<n$ bits.  A random hash function is a function $u$ chosen randomly and uniformly from the set of hash functions $U_m=\{u:\{0,1\}^n\rightarrow \{0,1\}^m\}$, such that $\forall \vs,\vx\in\{0,1\}^n,\vs\neq \vx:Pr[u(\vs)=u(\vx)]\leq 1/2^m$. If this property is satisfied by a sub-class $\bar{U}_m\subseteq U_m$ under uniform sampling, than $\bar{U}_m$ is called a {\em universal} class of hash functions~\cite{CarterJ:79}. An immediate fact about random hash functions from $U_m$ or from any universal sub-class $\bar{U}_m$ is that for any set of vectors $S\subset \{0,1\}^n$ with $|S|=s$ and a vector $\vx\notin S$, we have $Pr[\exists \vs\in S: u(\vs)=u(\vx)]\leq s/2^m$, which follows from the union bound. Note that this probability bound does not assume any probability distribution on the vectors $\vx,S$. 
\section{Compression rate vs. $p$-spread parameter}\label{sec:dmax}
In this section we seek coding schemes that given a parameter $p$ encode $\vy$ while knowing $\repZ$; the $p$-spread parameter $p'$ is known to the encoder from $\repZ$ and $p$. We investigate how the compression rate $|\enc(\vy,\repZ)|/\vlen$ depends on $p'$. We seek coding schemes that guarantee the reconstruction of any $\vy$ {\em without error}. The achievable compression rates we derive are given as average over the shared randomness between encoder and decoder, but we emphasize that unlike similar results in information theory, we do not allow any (even vanishing) decoding error, and we do not assume any stochastic model for $\vy$ or $\repZ$. \ycr{Formally, our coding schemes in this section operate over the following coding model.}
\ycr{
\begin{definition}\label{def:zero_error_av}
A coding scheme with parameters $p$, $p'$ has \textbf{zero-error average-rate} $R$ if for any choice of $\vy$ and $\repZ$ with $\pmax(\repZ,p)=p'$, $\vy$ can be uniquely recovered from $\enc(\vy,\repZ)$ and any $\vz\in\repZ$ s.t. $d_H(\vy,\vz)\leq p \vlen$, and $|\enc(\vy,\repZ)|/\vlen=R$ on average over randomness shared by the encoder and decoder.
\end{definition}
}
\ycrr{A useful subclass of Definition~\ref{def:zero_error_av} is {\em simple-hashing} zero-error average-rate coding schemes, which we define next. 

\begin{definition}\label{def:simple_hashing}
A zero-error average-rate coding scheme is called a \textbf{simple-hashing scheme} if with probability tending to $1$ (as $\vlen\rightarrow \infty$) it encodes $\vy$ as $u(\vy)$ such that for all $\vz\in\repZ$ and $\vy'\in\Hbl{p\vlen}{\vz}$, we have $u(\vy')\neq u(\vy)$ unless $\vy'=\vy$. The probability is taken over the drawings of $u(\boldsymbol{\cdot})\in U_m$, where $m$ is fixed given $p$, $p'$, $\vlen$.    
\end{definition}
Note that a simple-hashing scheme is free to encode $\vy$ arbitrarily with some (vanishing) probability, such that for every input it maintains the zero-error property of Definition~\ref{def:zero_error_av}.  
}
\subsection{Achievable rate with random hashing}\label{subsec:rand_bin}
In the first result we show a scheme in which the compression rate can be bounded by a simple function of $p$ and the $p$-spread parameter of the set of reference vectors $\repZ$.
\begin{theorem}\label{th:hash_y}
\ycrr{Given the parameters $p$ and $p'$,} there exists a \ycrr{simple-hashing} zero-error average-rate coding scheme with
\begin{equation} \lim_{\vlen\rightarrow \infty} \frac{|\enc(\vy,\repZ)|}{\vlen}\leq H(p) + H(p') + \epsilon, \label{eq:hash_y}\end{equation}
and $\epsilon >0$ is an arbitrary small real constant.
\end{theorem}
Before presenting the proof, we specify the encoder and decoder of the proposed coding scheme. The encoder and decoder share a random hash function from $U_m$ (e.g., by sharing random bits independent of the input), \ycrr{where $m$ is fixed and equal to $\vlen$ times the right-hand side of~\eqref{eq:hash_y}}. The scheme in fact works with any universal subclass of $U_m$, which by using known universal classes with structure can significantly reduce the number of bits shared by the encoder and decoder. In the following we use the definition
\[\repZ(\vx,\alpha) \triangleq \repZ \cap \Hbl{\alpha\vlen}{\vx}, \]
which is the set of reference vectors that are within distance $\alpha n$ from $\vx$.\\

\begin{cnst}\label{cnst:hash_y}
Let $u(\boldsymbol{\cdot})$ be a random hash function from $U_m$, where $m=\vlen[H(p) + H(p') + \epsilon]$. \\
\textbf{Encoder}: 1) List all reference vectors in $\repZ(\vy,p)$. 2) For each $\vz_i\in \repZ(\vy,p)$ apply the hash function $u$ on all vectors in $\Hbl{p\vlen}{\vz_i}$. In other words, apply $u$ on all vectors in $\cup_{\vz_i\in \repZ(\vy,p)}\Hbl{p\vlen}{\vz_i}$. 3) If no vector in these Hamming balls except $\vy$ is hashed to $u(\vy)$, output the bit $0$ followed by $u(\vy)$; otherwise output the bit $1$ followed by $\vy$.\\
\textbf{Decoder}: 1) If first bit is $1$, output the received $\vy$. If first bit is $0$, apply the hash function $u$ on all vectors in $\Hbl{p\vlen}{\vz}$ and output the unique vector whose hash equals the received $u(\vy)$.
\end{cnst}
\begin{proof}
Given $\vy$, by the problem statement the reference vector $\vz$ at the decoder satisfies $d_H(\vy,\vz)\leq p\vlen$. The encoder can list all vectors $\vz_i\in\repZ$ that satisfy $d_H(\vy,\vz_i)\leq p\vlen$. From the triangle inequality we get that if $\vz_i$ and $\vz_j$ are each at distance at most $p\vlen$ from $\vy$, then $d_H(\vz_i,\vz_j)\leq 2p\vlen$. From the $p$-spread parameter of $\repZ$ it follows that $d_H(\vz_i,\vz_j)\leq 2p'\vlen$. Hence the list of potential $\vz$ vectors given $\vy$ is an anticode with diameter $2p'\vlen$. It is known that the maximal size of an anticode with diameter $2p'\vlen$ is $|\Hbln{p'\vlen}|$~\cite{AhlswedeR:98}. Hence the set of vectors $\cup_{\vz_i}\Hbl{p\vlen}{\vz_i}$ hashed by the encoder has size bounded from above by $|\Hbln{p'\vlen}|\cdot |\Hbln{p\vlen}|\leq 2^{n[H(p')+H(p)]}$. From the properties of random hash functions, the probability that a vector in the set except $\vy$ will hash to $u(\vy)$ is at most $2^{-n\epsilon}$, going to zero as $\vlen$ grows. Hence the fraction of instances where the encoder outputs $u(\vy)$ tends to $1$. This gives    $|\enc(\vy)|\rightarrow m = \vlen\left[H(p) + H(p') + \epsilon \right]$ as $\vlen$ tends to infinity.
\end{proof}

The implication of Theorem~\ref{th:hash_y} is that knowing the set $\repZ$ at the encoder can improve the compression rate over known schemes when $p'<p$. For comparison, the scheme in~\cite{OrlitskyA:03} (which implicitly assumes the trivial $\repZ=\{0,1\}^{\vlen}$) gives $|\enc(\vy)|=nH(2p)$ with Gilbert-Varshamov non-explicit codes. Whenever $H(p) + H(p')<H(2p)$, Construction~\ref{cnst:hash_y} offers a better compression rate. \ycr{Note that Construction~\ref{cnst:hash_y} indeed fulfills the {\em zero-error average-rate} property of Definition~\ref{def:zero_error_av}: the average rate is bounded by \eqref{eq:hash_y} for the worst-case $\repZ$ given any $\vy$, and for any $\vz$ at the decoder.} \ycrr{Moreover, it is also a {\em simple-hashing scheme} because a fixed-$m$ $u(\cdot)$ provides unique decoding with probability tending to $1$.}\\
In practice, $\repZ$ may consist of reference vectors that are more ``favorable'' for compression than the cardinality upper bounds taken in the proof of Theorem~\ref{th:hash_y}. That means the benefits of knowing $\repZ$ at the encoder exceed the tighter compression-rate upper bounds presented in this paper.

\subsection{A converse result for random hashing}
The scheme of Construction~\ref{cnst:hash_y} encodes the input by random hashing of the vector $\vy$. \ycr{The next result shows that \ycrr{simple-hashing} zero-error average-rate coding schemes are subject to a fundamental lower bound on $|\enc(\vy,\repZ)|$.}
\begin{theorem}\label{th:lower_binning}
\ycr{Given the parameters $p$ and $p'$, any \ycrr{simple-hashing} zero-error average-rate coding scheme must have
\begin{equation}  \lim_{\vlen\rightarrow \infty} \frac{|\enc(\vy,\repZ)|}{\vlen} \geq H(p'+p). \label{eq:size_enc_bin_lb}\end{equation}}
\end{theorem}
\begin{proof}\ycr{First, since $m$ is fixed and $u(\vy)$ is the encoder output with probability tending to $1$, \eqref{eq:size_enc_bin_lb} is equivalent to the condition 
\begin{equation}  \lim_{\vlen\rightarrow \infty} \frac{m}{\vlen} \geq H(p'+p). \label{eq:size_hash_bin_lb}\end{equation}
\ycrr{By Definition~\ref{def:simple_hashing},} the encoder can output $u(\vy)$ only when there is no $\vy'\neq \vy$ within distance $p\vlen$ from $\vz\in\repZ$ such that $u(\vy')=u(\vy)$. In the proof we show that the probability over the functions $u(\boldsymbol{\cdot})\in U_m$ that no such $\vy'$ exists is vanishing with $\vlen$ if $m$ does not satisfy \eqref{eq:size_hash_bin_lb}. Given $\vy$, an adversary sets $\repZ=\Hbl{p'\vlen}{\vy}$ and examines all the vectors $\vy'\in \Hbl{(p'+p)\vlen}{\vy}$ and their hash values $u(\vy')$. If there exists a $\vy'$ with $u(\vy')=u(\vy)$, the adversary sets $\vz$ to be a vector in $\repZ=\Hbl{p'\vlen}{\vy}$ that is within distance $pn$ from $\vy'$; such a vector exists because $\vy'\in \Hbl{(p'+p)\vlen}{\vy}$, and as a result both $\vy$,$\vy'$ are within distance $p\vlen$ from $\vz\in\repZ$, as required. Asymptotically there are $\ssz\triangleq 2^{\vlen H(p'+p)}$ potential $\vy'$ vectors in $\Hbl{(p'+p)\vlen}{\vy}$. Denote $\rat=H(p'+p)$, and assume that $m$ violates~\eqref{eq:size_hash_bin_lb}, thus $\lim_{\vlen \rightarrow \infty}\vlen\rat-m=\infty$. Going over all the functions in $U_m$, there are $(2^m)^\ssz$ mappings from the vectors in $\Hbl{(p'+p)\vlen}{\vy}$ to the $2^m$ hash values. Out of these, there are $(2^m-1)^\ssz$ mappings in which all hash values are different from $u(\vy)$, which allow the encoder to successfully output $u(\vy)$. Taking the ratio between the   number of successful mappings and the total number of mappings, we get 

\begin{equation}
\frac{(2^m-1)^\ssz}{(2^m)^\ssz} = \left[1-2^{-m} \right]^{\ssz} = \left[1-2^{-m} \right]^{2^{\vlen \rat}} = \left(\left[1-2^{-m} \right]^{2^{m}}\right)^{2^{\vlen\rat - m}} \underset{\vlen\rightarrow \infty}{\longrightarrow} e^{-2^{\lim_{\vlen \rightarrow \infty} (\vlen\rat - m)}} \longrightarrow 0.
\end{equation}

Since the fraction of successful mappings of $\Hbl{(p'+p)\vlen}{\vy}$ is vanishing with $\vlen$, and uniformly drawing $u(\boldsymbol{\cdot})\in U_m$ induces a uniform distribution on these mappings, we proved that \eqref{eq:size_hash_bin_lb} is necessary to output $u(\vy)$ with non-vanishing probability, and \eqref{eq:size_enc_bin_lb} is necessary to output $u(\vy)$ with probability tending to $1$.}
 
\end{proof}

The gap between $H(p'+p)$ (Theorem~\ref{th:lower_binning}) and $H(p')+H(p)$ (Theorem~\ref{th:hash_y}) leaves room to potentially improve over Construction~\ref{cnst:hash_y} while still using \ycrr{simple} hashing. \ycrr{It is also possible that \eqref{eq:size_enc_bin_lb} can be improved by schemes that allow having $\vy$ and $\vy'$ with the same hash value, while finding a decoder that can somehow distinguish between the two hypotheses.}

\subsection{Reference-based coding}
The random-hashing scheme of Section~\ref{subsec:rand_bin} is attractive thanks to its simplicity. However, when the decoder knows the near neighborhood of its reference vector $\vz$ in $\repZ$, the following coding scheme may achieve smaller values of $|\enc(\vy,\repZ)|$. The idea of the next Construction~\ref{cnst:hash_z} is that hashing is done {\em not} on the input $\vy$, but on the reference vector in $\repZ$ nearest to $\vy$, which is used to encode $\vy$ along with a low-weight difference vector.
\begin{cnst}\label{cnst:hash_z}
Let $u(\boldsymbol{\cdot})$ be a random hash function from $U_m$, where $m=\vlen[2H(p') + \epsilon]$.\\
\textbf{Encoder}: 1) List all reference vectors in $\repZ(\vy,p)$. 2) Find in the list the vector nearest to $\vy$, denote it $\vz_1$ and define $d\triangleq d_H(\vy,\vz_1)$ and $\vv_1\triangleq \vy-\vz_1$ ($\vv_1$ is the difference vector between $\vy$ and $\vz_1$.) 3) For each $\vz_i\in \repZ(\vy,p)$ apply the hash function $u$ on all vectors $\vz_j\in \repZ(\vz_i,2p')$ such that $\vz_j+\vv_1\in \Hbl{p\vlen}{\vz_i}$. In other words, apply $u$ on all vectors in $\cup_{\vz_i\in \repZ(\vy,p)} [\repZ(\vz_i,2p')\cap \Hbl{p\vlen}{\vz_i-\vv_1}]$. 4) If none of these vectors except $\vz_1$ is hashed to $u(\vz_1)$, output the bit $0$ followed by $[u(\vz_1),\enum{d}{\vv_1}]$, where $\enum{d}{\vv_1}$ is the index of $\vv_1$ in an enumeration of $\Hbl{d}{\vzero}$ using $\vlen H\left(\frac{d}{\vlen}\right)$ bits; otherwise output the bit $1$ followed by $\vy$.\\
\textbf{Decoder}: 1) If first bit is $1$, output the received $\vy$. If first bit is $0$, apply the hash function $u$ on all vectors in $\repZ(\vz,2p')\cap \Hbl{p\vlen}{\vz-\vv_1}$, and for the unique vector $\vz_1$ whose hash equals $u(\vz_1)$, output $\vz_1+\vv_1$.
\end{cnst}
With the scheme in Construction~\ref{cnst:hash_z} we get the following result, obtained under the same assumptions of Theorem~\ref{th:hash_y}, that is, asymptotically as $\vlen\rightarrow \infty$ and on average over the random hash functions $u$.
\begin{theorem}\label{th:hash_z}
Let $\repZ$ be a set of reference vectors with $p$-spread parameter $p'$. Then there exists a zero-error average-rate coding scheme with
\begin{equation} \lim_{\vlen\rightarrow \infty}  \frac{|\enc(\vy,\repZ)|}{\vlen}\leq 2H(p') + \epsilon + H\left(\delta_H(\vy,Z)\right), \label{eq:size_hash_z}\end{equation}
where $\epsilon >0$ is an arbitrary small real constant and $\delta_H(\vy,Z)= d_H(\vy,Z)/\vlen$ is the fractional distance between $\vy$ and the nearest vector in $\repZ$.
\end{theorem}
\begin{proof}
We first note that the vectors $\vz_i\in \repZ(\vy,p)$ in part 3 of the encoder are all possible $\vz$ vectors at the decoder. In the proof of Theorem~\ref{th:hash_y} we already saw that there are at most $2^{\vlen H(p')}$ such vectors. Now for each $\vz_i$ considered as a possible $\vz$ vector at the decoder, the decoder does not know $\vz_1$, but knows that it is some $\vz_j\in \repZ(\vz_i,2p')$ (because both $\vz_i$,$\vz_1$ are at distance at most $p\vlen$ from $\vy$). We prove that for each $\vz_i\in Z(\vy,p)$ there are at most $2^{\vlen H(p')}$ vectors $\vz_j\in \repZ(\vz_i,2p')$ such that $\vz_j+\vv_1\in \Hbl{p\vlen}{\vz_i}$ (part 3 in the encoder); the latter property is required for $\vz_j$ to be consistent with $\vz_i$ at the decoder. To get this bound, observe that any pair $\vz_j,\vz_{j'}$ that both satisfy $\vz_j+\vv_1,\vz_{j'}+\vv_1\in \Hbl{p\vlen}{\vz_i}$ also satisfy $d_H(\vz_j,\vz_{j'})\leq 2p\vlen$, because both are in $\Hbl{p\vlen}{\vz_i-\vv_1}$. From the $p$-spread parameter this implies $d_H(\vz_j,\vz_{j'})\leq 2p'\vlen$. Now with the same argument as in the proof of Theorem~\ref{th:hash_y}, we upper bound by $|\Hbln{p'\vlen}|\leq 2^{\vlen H(p')}$ the number of vectors hashed in part 3 of the encoder for each $\vz_i\in Z(\vy,p)$. Having bounded by $2^{2\vlen H(p')}$ the union over all $\vz_i$ of $\vz_j$ vectors that may confuse the decoder given $\vz=\vz_i$, we conclude that a hash function with $\vlen (2H(p')+\epsilon)$ output bits is sufficient with probability $1-2^{-\vlen\epsilon}$ that tends to $1$. To complete the proof, we add to the encoder output an enumeration of the difference vector $\vv_1$, which can be done with $\vlen H\left(\delta_H(\vy,Z)\right)$ bits according to~\cite{CoverT:73}.

\end{proof}
\textbf{Discussion:} If $\vy$ is relatively close in Hamming distance to {\em any} vector in $\repZ$ (in particular not necessarily the $\vz$ at the decoder), then Construction~\ref{cnst:hash_z} allows to reduce the fractional encoding size from the $H(p) + H(p')$ of Theorem~\ref{th:hash_y} closer to $2H(p')$ in the first term of~\eqref{eq:size_hash_z}. In the worst case $\delta_H(\vy,Z)$ equals $p$, and then~\eqref{eq:size_hash_z} becomes $2H(p')+H(p)$, which is not competitive with the upper bound offered by Construction~\ref{cnst:hash_y}. However, with ``rich'' $\repZ$ sets many times the input $\vy$ would have a much closer $\vz_1$ vector. We have not been able to derive a converse result for reference-based coding. The core difficulty is to bound the advantage from the encoder's freedom to choose the reference vector in $\repZ$ (we do know how to get lower bounds when the encoder always uses the nearest vector as reference, like in Construction~\ref{cnst:hash_z}). \\
We add that it is easy to combine Constructions~\ref{cnst:hash_z} and~\ref{cnst:hash_y} such that the encoder chooses to hash $\vz_1$ when one is close to $\vy$, and $\vy$ itself when its near neighborhood in $\repZ$ is empty. This combination will only require another bit to mark to the decoder which of the constructions is used for each $\vy$.


\section{Fixed-Rate Compression with Low Complexity}\label{sec:low_complexity}
In addition to this paper's focus on having no statistical assumptions on the source and side information, in this section we aim to get schemes with {\em guaranteed worst-case} compression rates, and not just average rates with random hashing as in Section~\ref{sec:dmax}. We also return here to the more classical setup where the encoder does not have a list of possible reference vectors, so its knowledge is limited to the fact that the decoder's $\vz$ vector is at distance at most $p\vlen$ from the input $\vy$. In the terminology of Section~\ref{sec:model} we thus have $\repZ=\{0,1\}^{\vlen}$, and $\pmax(\repZ,p)=p$ (trivial $p$-spread parameter). This problem is classical and well studied, but our proposed schemes will allow to solve it efficiently even for long sequences, for example DNA sequences.
\subsection{Background: a known guaranteed fixed-rate scheme}
For the setup of compressing a length-$\vlen$ vector $\vy$ with an unknown $\vz\in\Hbl{p\vlen}{\vy}$ at the decoder,~\cite{OrlitskyA:03} proposed a coset-coding approach, where a length-$\vlen$ binary linear code with minimum distance $>2p\vlen$ is taken and used as follows.
\begin{cnst}\label{cnst:orlit}
\cite{OrlitskyA:03} Let $\code$ be a binary linear code with minimum distance $>2p\vlen$, and $\mS$ be a $\rho\vlen\times\vlen$ parity-check matrix for $\code$, $\rho\in(0,1)$.\\
\textbf{Encoder}: Given an input row vector $\vy$, calculate $\vs=\mS\vy^T$, and output $\vs$.\\
\textbf{Decoder}: Find the lowest-weight vector $\vv$ such that $\mS\vv^T=\vs+\mS\vz^T$; output $\vz+\vv$.
\end{cnst}
The output vector $\tilde{\vy}=\vz+\vv$ satisfies $\mS\tilde{\vy}^T=\vs$, like $\vy$, and having more than one such vector in $\Hbl{p\vlen}{\vz}$ would violate the minimum distance of $\code$. Hence $\tilde{\vy}=\vy$. This construction is guaranteed to succeed in recovering $\vy$ so long that indeed $d_H(\vz,\vy)\leq p\vlen$ as specified. In terms of complexity, the encoder of Construction~\ref{cnst:orlit} performs a matrix-vector product, with $\rho \vlen^2$ bit operations. The decoding complexity is much higher (equivalent to maximum-likelihood decoding of an error-correcting code); even if polynomial-time sub-optimal decoding is used, decoding complexity may be prohibitive for the values of $\vlen$ typical in applications like DNA sequences. Because of that issue, in the remainder of the section we develop guaranteed-decoding constructions that reduce decoding complexity by encoding the long sequence into a codeword composed of shorter sub-block codewords.

\subsection{Construction idea}
Our low-complexity constructions are based on the idea of {\em generalized concatenation (GC)}~\cite{BlokhZyablov:76}, adapted to the use of the codes for compression. As in GC, a long (length $\vlen$) binary vector is broken to much shorter (length $\blen$) sub-vectors, and non-binary outer codes encode a desired dependence among the sub-vectors. Different from GC, the encoder output is not a concatenated codeword, but only parity symbols of the outer codes. The key difference is that here for compression, the concatenation needs to design outer error-correcting codes for inner {\em coset} codes, and not inner error-correcting codes as usual. Moreover, to keep the decoding complexity below quadratic in $\vlen$, we design our codes with single-shot decoders for the outer codes. This is in contrast to the common use of GC constructions employing iterative decoders that decode up to half the minimum distance~\cite{ZyablovV:75}, building on the generalized minimum distance (GMD) method~\cite{ForneyD:66}. The particular sub-class of GC codes found useful here is {\em generalized error-locating (GEL)} codes~\cite{BossertBook:99}, because their construction through inner syndromes fits well the syndrome method of Construction~\ref{cnst:orlit}.

\subsection{First efficient construction}
We first define a partition of length-$\vlen$ vectors to $\vbr\triangleq\vlen/\blen$ sub-vectors of length $\blen$ each, where $\blen$ is some integer that divides $\vlen$. Thus for example $\vy=[\vy_1,\ldots,\vy_{\vbr}]$, where $,$ represents vector concatenation. Let $\{\mS^{(i)}\}_{i=1}^{m}$ be a set of binary matrices where $\mS^{(i)}$ has dimensions $r_i\times k$. For a sub-vector $\vy_j$ we further define the {\em partial $i$-th syndrome} as
\[ \vs_{j}^{(i)}=\mS^{(i)}\vy_j^T.  \]
$\vs_{j}^{(i)}$ is a column vector of dimension $r_i$. We take the matrices $\{\mS^{(i)}\}_{i=1}^{m}$ to be a nested set, meaning that for $i'>i$ the $r_i$ rows of $\mS^{(i)}$ appear in $\mS^{(i')}$ in concatenation with additional $r_{i'}-r_{i}$ rows. This implies that $r_i$ is increasing with $i$. When $\mS^{(i)}$ is seen as a parity-check matrix of a length-$\blen$ code $\code^{(i)}$, we denote its minimum distance by $d_i$. From the nesting property we know that $d_i$ is non-decreasing with $i$. We define the differential matrix $\dS^{(i)}$ to contain the rows in $\mS^{(i)}$ that do not appear in $\mS^{(i-1)}$, and the number of rows in $\dS^{(i)}$ is denoted $\dr_i\triangleq r_{i}-r_{i-1}$. For these definitions, $\mS^{(0)}$ is defined as the empty matrix, hence $\dS^{(1)}=\mS^{(1)}$ (and $\dr_1=r_1$). Define also $\mI_{b}$ as the identity matrix of order $b$. Our first concatenated construction now follows.


\begin{cnst}\label{cnst:concat}
Let $\{\code^{(i)}\}_{i=1}^{m}$ be a nested set of length-$\blen$ binary codes with parity-check matrices $\{\mS^{(i)}\}_{i=1}^{m}$ and minimum distances $\{d_i\}_{i=1}^{m}$. In addition, define the set $\{\mH^{(i)}\}_{i=2}^{m}$ where $\mH^{(i)}$ is a $\rho_i\times\vbr$ parity-check matrix over the finite field $F_{2^{\dr_i}}$ that defines a code with minimum distance $\delta_i$. Let $\enci^{(i)}:F_{2^{\dr_i}}^{\vbr}\rightarrow F_{2^{\dr_i}}^{\rho_i} $ be an encoder function mapping a length $\vbr$ vector over $F_{2^{\dr_i}}$ to the parity symbols of the code whose parity-check matrix is $[\mH^{(i)},\mI_{\rho_i}]$.  Define $\deci^{(i)}:F_{2^{\dr_i}}^{\vbr}\times F_{2^{\dr_i}}^{\rho_i}\rightarrow F_{2^{\dr_i}}^{\vbr} $ to be a decoder function with inputs $\va$,$\vb$ that finds the vector $\vx$ nearest to $\va$ that satisfies $\mH^{(i)}\vx^T=\vb$.\\
\textbf{Encoder}: Given an input row vector $\vy$:
\begin{enumerate}
\item Partition $\vy=[\vy_1,\ldots,\vy_{\vbr}]$.
\item Calculate $\vu_{j}^{(i)}:=\dS^{(i)}\vy_j^T$ for each $i\in\{1,\ldots,m\}$ and $j\in\{1,\ldots,\vbr\}$.
\item Encode $\vp^{(i)}:=\enci^{(i)}(\vu_{1}^{(i)},\ldots,\vu_{\vbr}^{(i)})$ for each $i\in\{2,\ldots,m\}$, and define $\vp^{(1)}:=[\vu_{1}^{(1)},\ldots,\vu_{\vbr}^{(1)}]$.
\item Output $\vp^{(i)}$, for every $i\in\{1,\ldots,m\}$.

\end{enumerate}
\textbf{Decoder}: Given encoder outputs $\vp^{(i)}$ and reference row vector $\vz$:
\begin{enumerate}
\item Partition $\vz=[\vz_1,\ldots,\vz_{\vbr}]$.
\item Initialize $\hat{\vs}_{j}^{(1)}:=\vu_{j}^{(1)}$, for each $j\in\{1,\ldots,\vbr\}$.
\end{enumerate}
\underline{Iterate on $i=2,\ldots,m$ in 3-5 below}:
\begin{enumerate}
\item[3)] For each $j$, find the lowest-weight vector $\vv_j$ such that $\mS^{(i-1)}\vv_j^T=\hat{\vs}_{j}^{(i-1)}+\mS^{(i-1)}\vz_j^T$.
\item[4)] Take $\hat{\vy}_{j}=\vz_j+\vv_j$ and calculate $$\hat{\vu}^{(i)}:=\deci^{(i)}(\dS^{(i)}\hat{\vy}_{1}^T,\ldots,\dS^{(i)}\hat{\vy}_{\vbr}^T;\vp^{(i)}).$$
\item[5)] Concatenate $\hat{\vs}_{j}^{(i)}:=[\hat{\vs}_{j}^{(i-1)};\hat{\vu}_{j}^{(i)}]$, for each $j\in\{1,\ldots,\vbr\}$.
\end{enumerate}
\underline{Output}:
\begin{enumerate}
\item[6)] For each $j$, find the lowest-weight vector $\vv_j$ such that $\mS^{(m)}\vv_j^T=\hat{\vs}_{j}^{(m)}+\mS^{(m)}\vz_j^T$.
\item[7)] Output $\hat{\vy}_{j}=\vz_j+\vv_j$, for each $j\in\{1,\ldots,\vbr\}$.
\end{enumerate}
\end{cnst}
In each iteration $i$ the decoder of Construction~\ref{cnst:concat} takes the partial syndromes $\hat{\vs}_{j}^{(i-1)}$ from the previous iteration, uses decoders for the (inner) code $\mS^{(i-1)}$ to find the nearest word to $\vz_j$ with partial syndrome $\hat{\vs}_{j}^{(i-1)}$, and then calculates the next differential syndromes $\dS^{(i)}\hat{\vy}_{j}^T$ of these nearest words. The iteration ends with correcting errors in the differential syndromes using the (outer) code $\mH^{(i)}$, and obtaining the next partial syndromes $\hat{\vs}_{j}^{(i)}$. The efficient realization of the steps in the encoder and decoder is discussed in Section~\ref{subsec:realization}.
\subsection{Code parameters for guaranteed decoding}
Construction~\ref{cnst:concat} needs to work with the only specification being that $d_H(\vz,\vy)\leq p\vlen$, that is, a distance bound for the full block. Then we specify parameters for the codes $\{\mS^{(i)}\}_{i=1}^{m}$ and $\{\mH^{(i)}\}_{i=2}^{m}$ that are sufficient for guaranteed decoding with Construction~\ref{cnst:concat}. The following lemma is the main tool for setting these parameters.

\begin{lemma}\label{lem:concat_params}
Let $d_H(\vz,\vy)\leq p\vlen$, and take Construction~\ref{cnst:concat} with parity-check matrices $\{\mS^{(i)}\}_{i=1}^{\reso}$ of binary codes with minimum distances $\{d_i\}_{i=1}^{\reso}$. Then correct decoding of $\vs_{j}^{(\reso)}$ by $\hat{\vs}_{j}^{(\reso)}$ is guaranteed if for each $i\in\{2,\ldots,\reso\}$ we use a parity-check matrix $\mH^{(i)}$ of a code with minimum distance $\delta_i > 4p\vlen/(d_{i-1}-1)$.
\end{lemma}
\begin{proof}
The basic observation is that  $d_H(\vz_j,\vy_j)>(d_{i-1}-1)/2$ can occur in less than $2pn/(d_{i-1}-1)$ of the indices $j\in\{1,\ldots,\vbr\}$. Since the previous inequality is necessary for $\hat{\vy}_{j}\neq \vy_{j}$ in step 4, the decoder $\deci^{(i)}$ will see less than $2p\vlen/(d_{i-1}-1)$ errors, and can correct them with distance  $\delta_i>4p\vlen/(d_{i-1}-1)$ for the code $\mH^{(i)}$. Recovering the correct $\vu_{j}^{(i)}$ for all $i,j$ guarantees that at every iteration $i$, $\hat{\vs}_{j}^{(i)}=\vs_{j}^{(i)}$, including in iteration $\reso$.
\end{proof}

Recall $\vlen=\blen\vbr$, and pick an integer $\reso$. For the matrix $\mS^{(i)}$ we specify the minimum distance
\begin{equation} d_i = 4p\blen+\frac{i}{\reso}\left(\frac{1}{2}-4p \right)\blen+1,~i\in\{1,\ldots,\reso-1\},\label{eq:concat_di}\end{equation}
and for $\mS^{(\reso)}$ we take a square full-rank matrix, hence $d_{\reso}=\infty$ meaning that the last code is the trivial code with just the all-zero codeword. Note that the $d_i$ from $i=1$ to $\reso-1$ form an affine progression between $4p\blen+1$ and $\frac{1}{2}\blen+1$ (not inclusive). For the parity-check matrices $\mH^{(i)}$ we define the corresponding distances to satisfy Lemma~\ref{lem:concat_params}
\begin{equation} \delta_i = \lfloor 4p\vlen/(d_{i-1}-1)\rfloor+1,~i\in\{2,\ldots,\reso\}.\label{eq:concat_deltai}\end{equation}

\textbf{Rate calculation}:
To calculate the compression rate of Construction~\ref{cnst:concat} we use the simple formula in the next lemma.
\begin{lemma}\label{lem:rate_concat}
The total number of bits output by the encoder of Construction~\ref{cnst:concat} is
\begin{equation}
\sum_{i=1}^{\reso}|\vp^{(i)}| =  r_{1}\vbr + \sum_{i=2}^{\reso}\dr_{i} \rho_{i}.\label{eq:concat_rate}
\end{equation}
\end{lemma}
\begin{proof}
Immediate from item 3 in the encoder of Construction~\ref{cnst:concat}.
$r_i$ is the redundancy of the binary code $\mS^{(i)}$ with minimum distance $d_i$, specifically $r_\reso=\blen$, and recall the definition $\dr_i=r_{i}-r_{i-1}$. $\rho_i$ is the redundancy of the $2^{\dr_i}$-ary code $\mH^{(i)}$ with minimum distance $\delta_i$.
\end{proof}
To get the asymptotic compression rate achievable with Construction~\ref{cnst:concat} we use the Gilbert-Varshamov bound for the binary codes $\mS^{(i)}$
\begin{equation} r_i = \blen H\left(\frac{d_i}{\blen}\right), \label{eq:r_i}\end{equation}
and the Singleton bound for the $2^{\dr_i}$-ary code $\mH^{(i)}$
\[ \rho_i = \delta_i.\]
The Singleton bound is achievable, e.g. with Reed-Solomon codes, when $2^{\dr_i}\geq t$. Since every $\dr_i$ grows linearly with $\blen$, for this condition to be met it is sufficient that $\blen$ is at least logarithmic in $\vbr=\vlen/\blen$, for example when $\blen = \log \vlen$.
Now we get the compression rate:
\begin{proposition}\label{prop:rate_concat}
For any constant integer $\reso$ the compression rate of Construction~\ref{cnst:concat}, which is the total number of bits output by the encoder divided by $\vlen$ is
\begin{equation}
H\left(4p+\frac{1}{\reso}\left(\frac{1}{2}-4p \right)\right)+\sum_{i=2}^{\reso} \left[H\left(4p+\frac{i}{\reso}\left(\frac{1}{2}-4p \right)\right)-H\left(4p+\frac{i-1}{\reso}\left(\frac{1}{2}-4p \right)\right)\right]\cdot\frac{4p}{4p+\frac{i-1}{\reso}\left(\frac{1}{2}-4p \right)}.\label{eq:concat_rate_exp}
\end{equation}
\end{proposition}
\begin{proof}
The expression in~\eqref{eq:concat_rate_exp} is obtained by substituting in~\eqref{eq:concat_rate} the Gilbert Varshamov bound for $r_i$ corresponding to $d_i$ in~\eqref{eq:concat_di}, and the Singleton bound for $\rho_i$ corresponding to $\delta_i$ in~\eqref{eq:concat_deltai}, then normalizing by $\vlen$.
\end{proof}

\subsection{Realization and complexity}\label{subsec:realization}
\subsubsection{Realization}
To realize Construction~\ref{cnst:concat} efficiently, we reduce encoding and decoding operations to known operations from error-correcting codes. Because error-correcting codes are used in a substantially different way for compression, we next explain their adaptations in the concatenated scheme.\\
\underline{The function $\enci^{(i)}:F_{2^{\dr_i}}^{\vbr}\rightarrow F_{2^{\dr_i}}^{\rho_i}$} calculates the parity symbols of the code with parity-check matrix $[\mH^{(i)},\mI_{\rho_i}]$, where $\mH^{(i)}$ is a parity-check matrix of a length-$\vbr$ code with minimum distance $\delta_i$, given in systematic form. The code $[\mH^{(i)},\mI_{\rho_i}]$ is a lengthened version of the code $\mH^{(i)}$. Note that this code is a poor error-correcting code, but works here (with better parameters) because there are no errors in the symbols of $\vp^{(i)}$. Given a systematic encoder for the code defined by $\mH^{(i)}$ (for example a Reed-Solomon code), we can realize $\enci^{(i)}$ by first encoding the first $\vbr-\rho_i$ input symbols to a word of $\mH^{(i)}$, and then subtracting from the $\rho_i$ parity symbols the remaining $\rho_i$ inputs. This guarantees a 1-1 mapping from length-$t$ input vectors to length-$(t+\rho_i)$ output vectors $\vc$ with $[\mH^{(i)},\mI_{\rho_i}]\vc^T=\vzero$. \\
\underline{The function $\deci^{(i)}:F_{2^{\dr_i}}^{\vbr}\times F_{2^{\dr_i}}^{\rho_i}\rightarrow F_{2^{\dr_i}}^{\vbr}$} needs to find the vector $\vx$ nearest to $\va$ that satisfies $\mH^{(i)}\vx^T=\vb$ ($\va$,$\vb$ are the first and second inputs to $\deci^{(i)}$, respectively). $\va$ is the vector of differential syndromes of the estimated $\hat{\vy}_j$ sub-vectors; $\vb$ is the output of $\enci^{(i)}$ that is available to the decoder without error. Given a syndrome decoder for the code $\mH^{(i)}$ (for example a Berlekamp-Massey Reed-Solomon decoder), $\deci^{(i)}$ can be implemented by invoking the decoder on the syndrome $\mH^{(i)}\va^T-\vb$, and subtracting the output minimal-weight error word from $\va$ to obtain $\vx$. This gives the desired output because we look for the minimal-weight $\ve$ such that $\vx=\va-\ve$ and $\mH^{(i)}\vx^T=\vb$, implying $\mH^{(i)}\ve^T=\mH^{(i)}\va^T-\vb$. Since $\vb$ is error-free, the correction capability of $\deci^{(i)}$ is the same as that of the syndrome decoder operating on the code $\mH^{(i)}$. \\
Another function needed in Construction~\ref{cnst:concat} appears in item 3 of its decoder: finding low-weight vectors with a given syndrome can be realized by known syndrome decoders for the codes $\mS^{(i)}$.

\subsubsection{Complexity}
Per the realizations above of the functions in Construction~\ref{cnst:concat}, we obtain the following encoding and decoding asymptotic complexities.\\
\underline{Decoding complexity}: for the codes $\mH^{(i)}$ we take Reed-Solomon codes over a field of size $\vbr$, which can be decoded with complexity $O(\vbr\log^2\vbr)$~\cite{JustesenJ:76}. For the codes $\mS^{(i)}$ we take binary linear codes that can be decoded with complexity at most $\blen\cdot 2^{\blen/2}$ using the trellis representation of the code (it is known~\cite{WolfJ:78} that every linear block code can be represented by a trellis with at most $2^{\min(r_i,\blen-r_i)}\leq 2^{\blen/2} $ states in each coordinate). Now taking $\blen = O(\log \vlen)$ we get the total complexity of $O(\vlen^{1.5})$, because for each block we invoke  $\vbr=\vlen/\log \vlen$ binary trellis decoders with total complexity
\[ O\left(\frac{\vlen}{\log\vlen}\sqrt{\vlen} \log\vlen \right) = O\left(\vlen^{1.5}\right).\]
The complexity of the Reed-Solomon decoders is asymptotically negligible compared to $O\left(\vlen^{1.5}\right)$ because we invoke a constant number $\reso$ of Reed-Solomon decoders, which give $O\left(\frac{\vlen}{\log\vlen}\log^2\frac{\vlen}{\log\vlen}\right)$ operations over finite-field elements represented as size $\alpha\log \vlen$ binary vectors (for some real $\alpha<1$), giving in total not more than $O\left(\vlen\log^3\vlen\right)$ bit operations.

Note that a construction using the standard generalized-concatenation half-minimum-distance decoder would have a higher complexity of $O\left(\vlen^{2}\log \vlen \right)$~\cite{ZyablovV:75}. Because $\vlen^2$ is considered prohibitive for long sequences, the decoder and parameters specified for Construction~\ref{cnst:concat} give a more practical alternative for realization.

\subsection{Improved construction}
To reduce the overall compression rate of the scheme, we now propose an improvement of Construction~\ref{cnst:concat} that still enjoys $O\left(\vlen^{1.5}\right)$ decoding complexity. The idea is that employing error-and-erasure decoding allows to set the correction parameters of the codes  $\mS^{(i)}$ and  $\mH^{(i)}$ such that less total redundancy is required. In the following improved scheme, we allow the decoder of $\mS^{(i)}$ to declare decoding failure when the distance of $\hat{\vy}_{j}$, the closest vector to $\vz_j$, is greater than $d_i/3$.

\begin{cnst}\label{cnst:concat2}
We repeat Construction~\ref{cnst:concat}, only changing the specification of $\deci^{(i)}$. Define $\deci^{(i)}:(F_{2^{\dr_i}}\cup *)^{\vbr}\times F_{2^{\dr_i}}^{\rho_i}\rightarrow F_{2^{\dr_i}}^{\vbr} $ to be a decoder function with inputs $\va$,$\vb$ that finds a vector $\vx$ that satisfies: 1) $\mH^{(i)}\vx=\vb$, and 2) $\vx$ is nearest to $\va$ on the subset of coordinates that are not $*$ in $\va$.\\
\textbf{Encoder}: same as Construction~\ref{cnst:concat}.\\
\textbf{Decoder}: Given an input row vector $\vz$:
\begin{enumerate}
\item Partition $\vz=[\vz_1,\ldots,\vz_{\vbr}]$.
\item Initialize $\hat{\vs}_{j}^{(1)}:=\vu_{j}^{(1)}$, for each $j\in\{1,\ldots,\vbr\}$.
\end{enumerate}
\underline{Iterate on $i=2,\ldots,m$ in 3-5 below}:
\begin{enumerate}
\item[3)] For each $j$, find the lowest-weight vector $\vv_j$ such that $\mS^{(i-1)}\vv_j^T=\hat{\vs}_{j}^{(i-1)}+\mS^{(i-1)}\vz_j^T$.
\item[4)] If the weight of $\vv_j$ is at most $(d_{i-1}-1)/3$, calculate $\hat{\vy}_{j}=\vz_j+\vv_j$ and take $\va_{j}:=\dS^{(i)}\hat{\vy}_{j}^T$; otherwise take $\va_{j}:=*$. Now calculate $$\hat{\vu}^{(i)}:=\deci^{(i)}(\va_{1},\ldots,\va_{t};\vp^{(i)}).$$
\item[5)] Concatenate $\hat{\vs}_{j}^{(i)}:=[\hat{\vs}_{j}^{(i-1)};\hat{\vu}_{j}^{(i)}]$, for each $j\in\{1,\ldots,\vbr\}$.
\end{enumerate}
\underline{Output}: same as Construction~\ref{cnst:concat}.
\end{cnst}

Now we specify parameters for the codes $\{\mS^{(i)}\}_{i=1}^{m}$ and $\{\mH^{(i)}\}_{i=2}^{m}$ that are sufficient for guaranteed decoding with Construction~\ref{cnst:concat2}. The following lemma is the modification of Lemma~\ref{lem:concat_params} to the improved construction.
\begin{lemma}\label{lem:concat_params2}
Let $d_H(\vz,\vy)\leq p\vlen$, and take Construction~\ref{cnst:concat2} with parity-check matrices $\{\mS^{(i)}\}_{i=1}^{m}$ of binary codes with minimum distances $\{d_i\}_{i=1}^{m}$. Then correct decoding of $\vs_{j}^{(m)}$ by $\hat{\vs}_{j}^{(m)}$ is guaranteed if for each $i\in\{2,\ldots,m\}$ we use a parity-check matrix $\mH^{(i)}$ of a code with minimum distance $\delta_i \geq 3pn/(d_{i-1}-1)$.
\end{lemma}
\begin{proof}
Denote by $\tau_{1}$ the number of indices $j\in\{1,\ldots,\vbr\}$ where $(d_{i-1}-1)/3<d_H(\vz_j,\vy_j)\leq 2(d_{i-1}-1)/3$, and by $\tau_{2}$ the number of indices where  $d_H(\vz_j,\vy_j)>2(d_{i-1}-1)/3$. From the global distance constraint it is implied that $\tau_{1}+2\tau_{2}<3pn/(d_{i-1}-1)$. The code $\mS^{(i-1)}$ has minimum distance $d_{i-1}$ and can thus simultaneously correct up to $(d_{i-1}-1)/3$ errors and detect up to $2(d_{i-1}-1)/3$ errors. Hence the decoder $\deci^{(i)}$ will see $\tau_{e}\leq \tau_{2}$ errors and $\tau_{*}=\tau_{1} +\tau_{2}-\tau_{e} $ erasures ($*$ symbols in Construction~\ref{cnst:concat2}). It is observed that $\tau_{*}+2\tau_{e}\leq \tau_{1}+2\tau_{2}<3pn/(d_{i-1}-1)$, and hence minimum distance of $\delta_i \geq 3pn/(d_{i-1}-1)$ is sufficient for the code $\mH^{(i)}$ to recover $\vu_{j}^{(i)}$ and in turn $\vs_{j}^{(i)}$ correctly.
\end{proof}
Recall $\vlen=\blen\vbr$, and pick an integer $\reso$. For the matrix $\mS^{(i)}$ we specify the minimum distance
\begin{equation} d_i = 3pk+\frac{i}{\reso}\left(\frac{1}{2}-3p \right)k+1,~i\in\{1,\ldots,\reso-1\},\label{eq:concat2_di}\end{equation}
and for $\mS^{(\reso)}$ we take a square full-rank matrix, hence $d_{\reso}=\infty$ meaning that the last code is the trivial code with just the all-zero codeword. Note that the $d_i$ from $i=1$ to $\reso-1$ form an affine progression between $3pk+1$ and $\frac{1}{2}k+1$ (not inclusive). For the parity-check matrices $\mH^{(i)}$ we define the corresponding distances
\begin{equation} \delta_i = \lceil 3pn/(d_{i-1}-1)\rceil,~i\in\{2,\ldots,\reso\}.\label{eq:concat2_deltai}\end{equation}

To get the asymptotic compression rate achievable with Construction~\ref{cnst:concat2} we adjust Proposition~\ref{prop:rate_concat} to the $d_i$ and $\delta_i$ of the improved construction.
\begin{proposition}\label{prop:rate_concat2}
For any constant integer $\reso$ the compression rate of Construction~\ref{cnst:concat2}, which is the total number of bits output by the encoder divided by $\vlen$ is
\begin{equation}
H\left(3p+\frac{1}{\reso}\left(\frac{1}{2}-3p \right)\right)+\sum_{i=2}^{\reso} \left[H\left(3p+\frac{i}{\reso}\left(\frac{1}{2}-3p \right)\right)-H\left(3p+\frac{i-1}{\reso}\left(\frac{1}{2}-3p \right)\right)\right]\cdot\frac{3p}{3p+\frac{i-1}{\reso}\left(\frac{1}{2}-3p \right)}.\label{eq:concat2_rate_exp}
\end{equation}
\end{proposition}
\begin{proof}
The expression in~\eqref{eq:concat2_rate_exp} is obtained by substituting in~\eqref{eq:concat_rate} the Gilbert Varshamov bound for $r_i$ corresponding to $d_i$ in~\eqref{eq:concat2_di}, and the Singleton bound for $\rho_i$ corresponding to $\delta_i$ in~\eqref{eq:concat2_deltai}, then normalizing by $\vlen$.
\end{proof}
We plot in Fig.~\ref{fig:rt_compare} the resulting compression rates of Construction~\ref{cnst:concat} (dashed) and Construction~\ref{cnst:concat2} (solid), as a function of $p$, in the range $p\in [0,2.5\cdot 10^{-3}]$; the plots evaluate the expressions in~\eqref{eq:concat_rate_exp},~\eqref{eq:concat2_rate_exp}, respectively, with $\reso=20000$. We do not compare these rates to the better rates of the basic Construction~\ref{cnst:orlit} ($H(2p)$ assuming error-correcting codes meeting the Gilbert Varshamov bound), because of its exponential decoding complexity. A more relevant comparison is with practical DNA compression algorithms, which currently give compression rates in the range $[0.1,0.2]$ (where the lower rates are achieved by algorithms with reference at both the encoder and decoder)~\cite{BonfieldJ:13}. We conclude that for the range plotted in Fig.~\ref{fig:rt_compare}, Construction~\ref{cnst:concat2} gives rates competitive with the state-of-the-art in DNA compression, and without need to use a reference at the encoder.
\begin{figure}[htbp]
\centerline{\includegraphics[width=0.6 \textwidth]{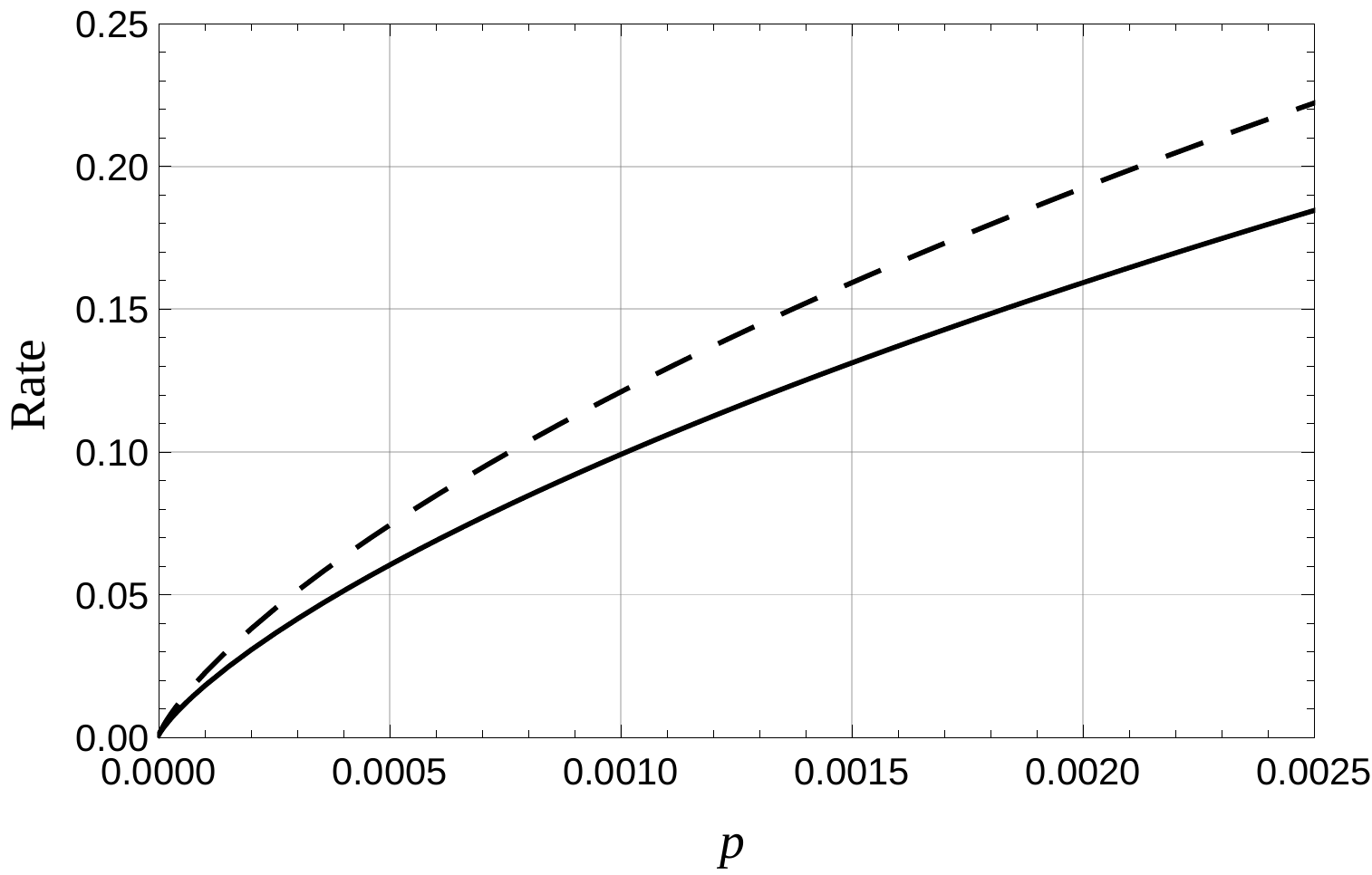}}
\caption{Compression rates of Construction~\ref{cnst:concat} (dashed) and Construction~\ref{cnst:concat2} (solid) as a function of the distance fraction $p$.}
\label{fig:rt_compare}
\end{figure}

\section{Conclusion}
The first part of the paper refines the classical problem of compression with side information using a combinatorial characterization of the size-information vectors $\repZ$. In addition to the $p$-spread parameter investigated here, it is interesting in future work to study compressibility with respect to other characterizations of $\repZ$. For example, instead of the $\max$ in~\eqref{eq:dmax}, one can characterize $\repZ$ by the full {\em spectrum} of distances in $\repZ$. The second part of the paper develops a concatenated scheme for efficient guaranteed compression with Hamming-bounded side information. A natural future work is to extend the scheme to also allow side information with insertions and deletions. While for long blocks insertions and deletions are notoriously difficult to handle, the short inner codes of the concatenated scheme may enable an efficient solution.        
\section{Acknowledgement}
We thank Neri Merhav for valuable discussions. We also thank the anonymous reviewers for valuable comments. This work was supported in part by the US-Israel Binational Science Foundation and in part by the Israel Science Foundation.
\bibliographystyle{IEEEtran}
\bibliography{myrefs1}

\end{document}